\documentclass[letterpaper, 10 pt, conference]{ieeeconf}  

\IEEEoverridecommandlockouts                              

\usepackage{cite}

\usepackage{amsmath,amssymb,amsfonts,amsthm}
\usepackage{algorithm}
\usepackage{algorithmic}
\usepackage{graphicx}
\usepackage{textcomp}
\usepackage{xcolor}

\usepackage{tikz}
\usetikzlibrary{positioning}
\usepackage{circuitikz}
\usepackage{pstricks}
\usepackage{float}
\usepackage{framed}

\usepackage[subtle,tracking=normal]{savetrees}

\newtheorem{theorem}{Theorem}

\newtheorem{prop}[theorem]{Proposition}

\title{\LARGE \bf
Towards Energysheds: A Technical Definition and Cooperative Framework for Future Power System Operations 
}

\author{Dakota Hamilton, Samuel Chevalier, Amritanshu Pandey, and Mads Almassalkhi$^{1}$
\thanks{This material is based upon work supported by the U.S. Department of Energy’s Office of Energy Efficiency and Renewable Energy (EERE) under the Enabling Place-Based Renewable Power Generation using Community Energyshed Design initiative, award number DE-EE0010407. The views expressed herein do not necessarily represent the views of the U.S. Department of Energy or the United States Government.}
\thanks{$^{1}$ D. Hamilton, S. Chevalier, A. Pandey, and M. Almassalkhi are with the Department of Electrical and Biomedical Engineering, University of Vermont, Burlington, VT, USA
{\tt\small \{dhamilt6, schevali, armitanshu.pandey, malmassa\}@uvm.edu}}%
}

\begin{document}

\maketitle
\thispagestyle{empty}
\pagestyle{empty}

\begin{abstract}

There is growing interest in understanding how interactions between system-wide objectives and local community decision-making will impact the clean energy transition. 
The concept of energysheds has gained traction in the areas of public policy and social science as a way to study these relationships. 
However, development of technical definitions of energysheds that permit system analysis are still largely missing. 
In this work, we propose a mathematical definition for energysheds, and introduce an analytical framework for studying energyshed concepts within the context of future electric power system operations. 
This framework is used to develop insights into the factors that impact a community's ability to achieve energyshed policy incentives within a larger connected power grid, as well as the tradeoffs associated with different spatial policy requirements. 
We also propose an optimization-based energyshed policy design problem, and show that it can be solved to global optimality within arbitrary precision by employing concepts from quasi-convex optimization.
Finally, we investigate how interconnected energysheds can cooperatively achieve their objectives in bulk power system operations.
\end{abstract}

\section{Introduction}

Public policies to decarbonize are going hand-in-hand with scaling up the integration of (distributed) renewable generation and electrification programs. These policies are transforming the generation, heating, cooling and transportation sectors nationally and globally. However, this global transformation depends crucially on public buy-in from regional and local communities, because top-down policies and incentives require bottom-up participation to succeed. 
For example, communities are very interested in how local renewable generation can benefit local economies. Local utilities in the U.S. are working with communities to better understand how shared community solar, net-metering, and energy technologies can avoid the need for expensive grid upgrades~\cite{gmpTime}. 
In Europe, the formation of local energy communities is becoming a popular solution for enabling green energy transition~\cite{NordicLEC,SwissLEC}.
Clearly, there is an interplay between local decisions, regional costs, and global policy objectives and impacts~\cite{kristov2016tale}. 

Furthermore, the interactions between local communities, regional utilities, and national policy-makers highlight the importance of effective coordination.
This is partly why the U.S. Department of Energy has recently focused programs around studying Clean Energy Communities and EnergySheds~\cite{CE2C,doeSetoProgram}.


The concept of an \textit{energyshed} is similar to that of a watershed or foodshed~\cite{evarts2016,derolph2019,thomas2021}. While a watershed inherently limits the sharing of its resources since it is defined by the inherent (static) gradients of nature's surfaces and reservoirs (storage), a foodshed permits sharing its local food production (resources) with other foodsheds~\cite{schreiber2021quantifying}. The sharing between foodsheds is enabled exactly by transportation networks that interconnect them. This sharing among foodsheds also means that not all food need to be produced locally, and blurs the physical boundary of any individual foodshed. It is in this context that we present energysheds as local energy communities within (blurry) geographical areas in which energy system objectives and constraints are determined and between which energy can be actively exchanged.

These interactions also beget interesting questions about how energyshed decisions (or constraints) affect the power systems that interconnect them all. Conversely, the power systems within which all energysheds reside and the gradients associated with the physical power flows may impose constraints on what each energyshed can achieve locally. For example, an energyshed at the end of a long radial feeder may not have the network capacity available to meet 100\% of local demand with local (clean) generation over a day due to export limits (caused by transformer or voltage bounds). Thus, as far as the authors are aware, this paper is the first attempt to analyze and better understand the role of the power network in enabling or limiting local energyshed objectives.  Specifically, we propose a first mathematical definition of an energyshed, and study fundamental tradeoffs associated with energyshed decisions and interconnections and how power systems can enable or limit energyshed objectives.

\noindent To this end, the key contributions of this work are:
\begin{enumerate}
    \item A mathematical definition for energysheds, considering spatiotemporal and techno-economic facets.
    \item Theoretical analysis to provide insights into how limited resources and power system constraints impact the ability of individual communities to meet policy goals.
    \item  An optimization-based framework  for cooperative energyshed policy, which helps ensure equitable outcomes as communities invest in local generation resources.
    \item Computationally tractable methods, leveraging concepts from quasi-convex optimization, for solving the proposed optimization problem to global optimality (to within arbitrary precision).
\end{enumerate}


The rest of the paper is organized as follows. A technical definition of an energyshed is introduced in Sec.~\ref{sec:eshed_def}, and theoretical insights associated with energyshed operations are discussed. In Sec.~\ref{sec:framework}, a mathematical framework for studying cooperative energyshed policy design is proposed. Numerical case studies are presented in Sec.~\ref{sec:results}, and Section~\ref{sec:conclusion} concludes and discusses future research directions.

\section{Energyshed Definition and Analysis}
\label{sec:eshed_def}



Consider a power system represented as a graph, $\mathcal{G} := (\mathcal{N},\mathcal{E})$, where $\mathcal{N}$ is a set of nodes and $\mathcal{E}$ the set of edges, such that $(i,j) \in \mathcal{E}$ if nodes $i,j \in \mathcal{N}$ are connected by an edge. Denote the time window of interest, $T$, by a set of discrete time intervals $\mathcal{T}$. Then, each node $i$, at time $t$, has local demand, $P_{i,t}^\text{L}\ge 0$, local generation, $P_{i,t}^\text{G}\ge 0$, and net generation, $P_{i,t}^\text{G}-P_{i,t}^\text{L}$.

Furthermore, consider a community $k$ consisting of a contiguous set of nodes $i\in \mathcal{N}_k \subseteq \mathcal{N}$.\footnote{Formally, this means that the subgraph induced in $\mathcal{G}$ by $\mathcal{N}_k$ is connected.}
We introduce a desired ratio of local energy produced to local energy consumed, denoted~$\mathcal{X}$, as a metric.
We choose to define energysheds with a ratio (as opposed to a difference between local demand and generation) since public policy objectives are often based on percentages (e.g., see Vermont's renewable energy portfolio standard~\cite{vtCompEnergyPlan}).
Thus, we propose the following definition of an energyshed:

\begin{framed}
    \raggedright
  \textbf{Proposed Definition (Energyshed)}:   
  
  Energy community $k$ defined by the spatiotemporal 3-tuple $(\mathcal{X}_k,\:\mathcal{N}_k,\:T)$ represents an \textit{energyshed}, if it satisfies
\begin{equation}
\label{eq:chi_init}
    \mathcal{X}_k \le  
    \frac{\sum_{i\in\mathcal{N}_k} \sum_{t\in\mathcal{T}} P^\text{G}_{i,t}}{\sum_{i\in \mathcal{N}_k} \sum_{t\in\mathcal{T}} P^\text{L}_{i,t}}\,,
\end{equation}
for \textit{every} non-overlapping interval~$\mathcal{T}$ of duration~$T$.
\end{framed}

For example, if $\mathcal{N}_k$ represents the U.S. state of Vermont's entire power network, $T$ represents every year with $\mathcal{X} = 0.9$, then~\eqref{eq:chi_init} states that the state of Vermont becomes an energyshed when it meets 90\% of annual energy needs from local generation.
In another extreme, consider a 3-tuple representing a set of nodes that make up an islanded microgrid, $T$ being every second, and $\mathcal{X}=1$ (with no load shedding, as expected when islanded); then a microgrid can also be cast as an energyshed. Between an entire state over a year and a small microgrid over a second is a rich set of relevant power systems that underpin diverse energy communities as possible energysheds. Besides the ability to use~\eqref{eq:chi_init} to measure how far a given community is from meeting the requirements of an energyshed (i.e., backcasting historical data), we are also interested in understanding what a community needs to or can do to meet their own requirement(s), and how multiple interconnected communities can satisfy their individual energyshed requirements with or without sharing energy. This is discussed next.




\subsection{Energyshed resource limits}

In the following analysis, we consider each community to embody a single node in the network, since networks can be reduced via nodal aggregations and Kron-based reductions (e.g., see~\cite{chevalier2022}). Thus, in addition to dropping the sum over $\mathcal{N}_k$,  the 3-tuple defining community $k$ becomes $(\mathcal{X}_k\,,~k\,,~T)$. 
In making this simplification, we assume that constraints internal to the community (e.g., internal line flow limits) are not binding.

Consider that a community has access to certain resources, which represent the capability of said community to invest in energy assets, including (but not limited to) financial, technological, or land use capabilities (e.g., transformer upgrades, load control programs, electrification incentives, batteries, electric vehicle charging stations, solar photovoltaic arrays, wind farms, etc.). 
Investment in such energy assets may enable a community to increase its local generation ratio by providing additional operational flexibility.
However, the resources of each community $k$ are subject to a certain budget, which limits the community's operational capacity to modify its net generation. Therefore, we assume that community $k$ at time $t$ can modify its net operational flexibility by $P_{k,t}^{\text{S}} = P_{k,t}^{\text{S}+} - P_{k,t}^{\text{S}-}\in [-\overline{P}_{k,t}^{\text{S}-}, \overline{P}_{k,t}^{\text{S}+}]$
, where $P_{k,t}^{\text{S}+}, P_{k,t}^{\text{S}-}\ge 0$.

The goal of community $k$ is then to understand how to best use its available capacity to maximize $\mathcal{X}_k$ (i.e., achieve its energyshed objectives), which now depends on $P_{k,t}^{\text{S}}$ as follows:\footnote{We consider energysheds with $\sum_{t\in\mathcal{T}} (P_{k,t}^\text{L} - P_{k,t}^\text{G}) > 0$ before modifying their net operational flexibility such that $\mathcal{X}_k < 1$.}
\begin{equation}
\label{eq:chi_def}
    \mathcal{X}_k = \frac{\sum_{t\in\mathcal{T}} (P^\text{G}_{k,t}+P^{\text{S}+}_{k,t})}{\sum_{t\in\mathcal{T}} (P^\text{L}_{k,t}+P^{\text{S}-}_{k,t})}\,,
\end{equation}
where the additional power generated at time~$t$ by energy assets (e.g., rooftop or utility solar, wind, batteries discharging) in community~$k$ is represented by $P^{\text{S}+}_{k,t} \ge 0$, and $P^{\text{S}-}_{k,t}\ge 0$ denotes additional power demand due to flexible loads (e.g., battery charging, electric vehicles, heating, and cooling loads). Clearly, in~\eqref{eq:chi_def}, we treat $P^\text{G}_{k,t}$ and $P^\text{L}_{k,t}$ as fixed (historical or predicted) data, whereas $P^{\text{S}+}_{k,t}$ and $P^{\text{S}-}_{k,t}$ are treated as variable operational decisions made by the community in pursuit of policy goals.



Thus, given representative data for demand and existing (non-dispatchable) generation, we can determine an analytical relationship between these operational capacity budgets and the maximum value of $\mathcal{X}_k$ that can be achieved by a community.
These relationships depend on the ability of a community to export excess generation to surrounding communities through the  grid infrastructure, which is discussed next.

\subsection{Communities with unconstrained power exports}
\label{sec:uncon_power_sharing}


First, we consider communities that can readily export power to the grid during instances in time when local power generation exceeds local power demand. 
Thus, we seek to determine the maximum $\mathcal{X}_k$, i.e.,  $\overline{\mathcal{X}_k}$,  that a community can achieve given a certain operational capacity, $[-\overline{P}_{k,t}^{\text{S}-}, \overline{P}_{k,t}^{\text{S}+}]$.

\begin{prop}[Maximum ratio with unconstrained export]
    Given representative $P^\text{G}_{k,t}$ and $P^\text{L}_{k,t}$ and operational capacity $[-\overline{P}_{k,t}^{\text{S}-}, \overline{P}_{k,t}^{\text{S}+}]$ and no limit on $P^\text{G}_{k,t} + P^\text{S}_{k,t} - P^\text{L}_{k,t}$,  the maximum $\mathcal{X}_k$ is given by $ \overline{\mathcal{X}}_k  = \mathcal{X}_k^0 + \frac{1}{\gamma_k} \sum_{t\in \mathcal{T}} \overline{P}^{\text{S}+}_{k,t}$, where $\mathcal{X}_k^0 := \frac{\sum_{t\in\mathcal{T}} P^\text{G}_{k,t}}{\sum_{t\in\mathcal{T}} P^\text{L}_{k,t}}$ and $\gamma_k :={\sum_{t\in\mathcal{T}} P^\text{L}_{k,t}}$.
\end{prop}
\begin{proof}
Since all quantities in~\eqref{eq:chi_def} are positive, we can maximize $\mathcal{X}_k$ by maximizing the numerator and minimizing the denominator. This is achieved when
$P^{\text{S}+}_{k,t} = \overline{P}^{\text{S}+}_{k,t}$ and $ P^{\text{S}-}_{k,t} = 0$ for all times $t$.
Thus, we obtain the maximum ratio of local generation to local consumption,
    \begin{equation} \label{eq:max_chi_sharing}
        \overline{\mathcal{X}}_k 
        = \frac{\sum_{t\in\mathcal{T}} P^\text{G}_{k,t}}{\sum_{t\in\mathcal{T}} P^\text{L}_{k,t}} \,+\, \frac{\sum_{t\in\mathcal{T}} \overline{P}^{\text{S}+}_{k,t}}{\sum_{t\in\mathcal{T}} P^\text{L}_{k,t}}\,.
    \end{equation}
    \vspace{-2mm}
\end{proof}

Clearly, without constraints on net power exports, there is linear relationship between the maximum local generation ratio, $\overline{\mathcal{X}}_k$, and the total energy capacity budget, $\sum_{t\in\mathcal{T}} \overline{P}^{\text{S}+}_{k,t}$, as shown in Fig.~\ref{fig:max_chi_single}.
Note that the slope of this line is $\frac{1}{\gamma_k}$ 
for community~$k$.
That is, communities with higher energy demand will require more operational capacity to meet the same energyshed policy goals (i.e., the same value of $\mathcal{X}_k$). 

Additionally, note that the $y$-intercept of
~\eqref{eq:max_chi_sharing} is $\mathcal{X}_k^0$, which
is the ratio of energy produced by local base generation to total demand over interval $\mathcal{T}$.
This also means that communities with existing local generation will require less operational capacity to meet their energyshed objectives.

\begin{figure}
    \centering
    \begin{tikzpicture}
        \footnotesize
        \fontfamily{ptm}\selectfont
        \node[inner sep=0pt] at (0.0,0.0)
        {\includegraphics[width=0.48\textwidth,trim=0 0 0 0,clip=true]{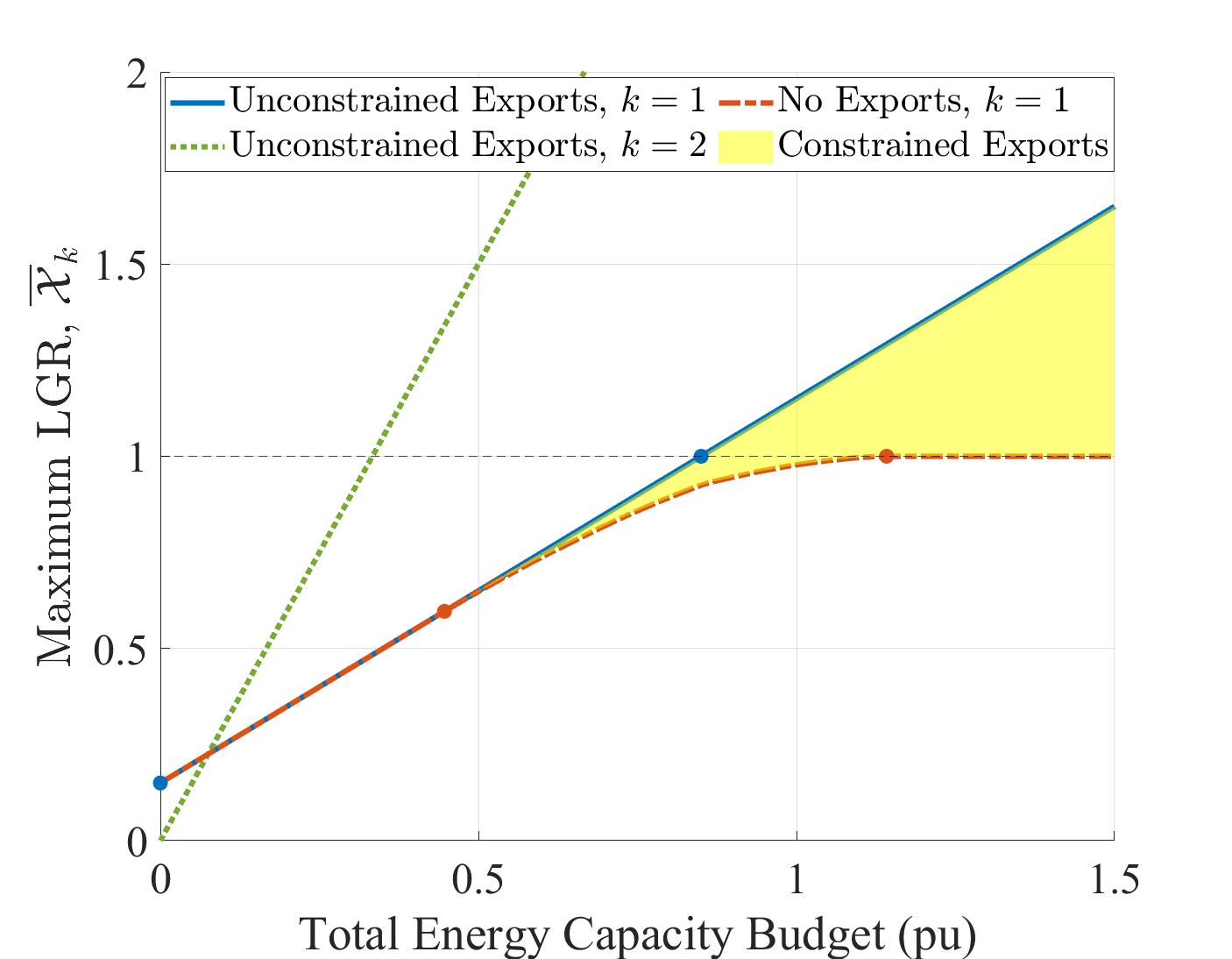}};
        \node[] () at (0.25,-2.25) {Existing ratio of local energy generation, $\mathcal{X}_1^0$};
        \draw[-latex] (-2.375,-2.25) -- (-3,-2.125);
        \node[align=center] () at (0.5,-1.55) {Local power generation exceeds \\ demand at some time $t$};
        \node[align=center] () at (-1.375,-0.7) {(A)};
        \draw[-latex] (-0.875,-1.25) -- (-1.125,-1.0);
        \node[align=center] () at (1.75,-0.625) {Local power generation meets \\ demand at all times $t$};
        \node[align=center] () at (1.875,0.375) {(B)};
        \draw[-latex] (1.75,-0.375) -- (1.875,0.125);
        \node[align=center] () at (0.375,1.0) {Local energy generation \\ meets demand};
        \draw[-latex] (0.25,0.625) -- (0.5,0.25);
        \node[align=center] () at (-2.25,0.75) {Slope = $\gamma_2$};
        \draw[-latex] (-2.25,0.62) -- (-1.75,0.25);
        \end{tikzpicture}
    
    \caption{Relationship between maximum achievable local generation ratio and total energy capacity budget under different power export constraints for two energysheds $k$. The total energy capacity budget is given in per unit using the total energy demand, $\sum_{t\in\mathcal{T}} P^{\text{L}}_{k,t}$ , as a base.
    }
    \label{fig:max_chi_single}
\end{figure}
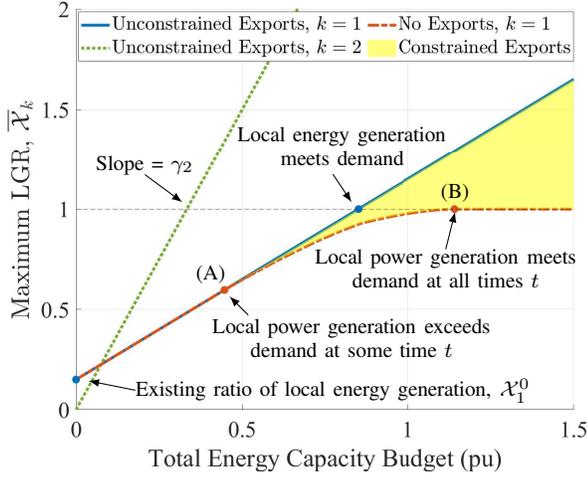

\subsection{Communities with constrained power exports}
\label{sec:con_pow_share}

In practice, an energyshed's ability to export power may be constrained by a number of factors including capacity limits of physical grid infrastructure (e.g., substation transformers, hosting capacity, transmission or distribution interconnections) or the willingness of surrounding communities to consume that power.
We model these constraints as bounds on the net export generation. That is,
\begin{equation}
\label{eq:Ps_con}
    \underline{P}^{\text{S}}_{k,t} \le P^{\text{S}+}_{k,t} - P^{\text{S}-}_{k,t} \le \overline{P}^{\text{S}}_{k,t}\,.
\end{equation}
Note that since $P^\text{G}_{k,t}$ and $P^\text{L}_{k,t}$ are fixed base generation and load data, constraints of the form~\eqref{eq:Ps_con} are sufficient to capture limits on net exports from community~$k$. 

\begin{prop}[Maximum ratio with constrained export]
    Given representative $P^\text{G}_{k,t}$ and $P^\text{L}_{k,t}$ and operational capacity $[-\overline{P}_{k,t}^{\text{S}-}, \overline{P}_{k,t}^{\text{S}+}]$ and constrained net exports $P^\text{G}_{k,t} + P^\text{S}_{k,t} - P^\text{L}_{k,t} \le \overline{P}_{k,t}^{S}$,  the maximum $\mathcal{X}_k$ is given by 
    \begin{equation} \label{eq:max_chi_no_sharing}
            \overline{\mathcal{X}}_k = \frac{\sum_{t\in\mathcal{T}} (P^\text{G}_{k,t}+\overline{P}^{\text{S}+}_{k,t})}{\sum_{t\in\mathcal{T}} (P^\text{L}_{k,t}+\max\{\overline{P}^{\text{S}+}_{k,t}-\overline{P}^{\text{S}}_{k,t}\,,0\})}\,.
   \end{equation}
\end{prop}
\begin{proof}
Following the same logic as in Sec.~\ref{sec:uncon_power_sharing}, we seek to maximize $P^{\text{S}+}_{k,t}$ and minimize $P^{\text{S}-}_{k,t}$ to maximize $\mathcal{X}_k$. Now, for some time~$t$, consider two cases:

\noindent
\textit{Case~I:} If $\overline{P}^{\text{S}}_{k,t} \ge \overline{P}^{\text{S}+}_{k,t}$ , then $P^{\text{S}+}_{k,t} = \overline{P}^{\text{S}+}_{k,t}$ and $P^{\text{S}-}_{k,t} = 0$ are maximizing $\mathcal{X}_k$ as in Sec.~\ref{sec:uncon_power_sharing}.

\noindent
\textit{Case~II:} If $\overline{P}^{\text{S}}_{k,t} < \overline{P}^{\text{S}+}_{k,t}$, then~\eqref{eq:Ps_con} is binding. Therefore, assuming $\sum_{t\in\mathcal{T}}\overline{P}^{\text{S}}_{k,t} \ge \sum_{t\in\mathcal{T}} (P^{\text{L}}_{k,t}-P^{\text{G}}_{k,t})$,   $\mathcal{X}_k$ is maximized by setting $P^{\text{S}-}_{k,t} = P^{\text{S}+}_{k,t} -  \overline{P}^{\text{S}}_{k,t}$.\footnote{Note that if $\sum_{t\in\mathcal{T}}\overline{P}^{\text{S}}_{k,t} < \sum_{t\in\mathcal{T}} (P^{\text{L}}_{k,t}-P^{\text{G}}_{k,t})$, then $\mathcal{X}_k$ is maximized by setting $P^{\text{S}+}_{k,t} = \overline{P}^{\text{S}}_{k,t}$ and $P^{\text{S}-}_{k,t} = 0$. However, under these conditions, $\mathcal{X}_k \ge 1$ for the proposed strategy. Since we are primarily interested in communities where $\mathcal{X}_k < 1$, we choose to ignore this case.}
Thus, to maximize $\mathcal{X}_k$ we set $P^{\text{S}+}_{k,t} = \overline{P}^{\text{S}+}_{k,t}$ and $P^{\text{S}-}_{k,t} = \overline{P}^{\text{S}+}_{k,t} -  \overline{P}^{\text{S}}_{k,t}$.
Combining the two above cases ($\overline{P}^{\text{S}}_{k,t} \ge \overline{P}^{\text{S}+}_{k,t}$ and $\overline{P}^{\text{S}}_{k,t} < \overline{P}^{\text{S}+}_{k,t}$), we obtain maximum $\overline{\mathcal{X}}_k$.
\end{proof}


Examples of the relationship between $\overline{\mathcal{X}}_k$ and total energy capacity budget, $\sum_{t\in\mathcal{T}} \overline{P}^{\text{S}+}_{k,t}$, for different power export constraints are shown in Fig.~\ref{fig:max_chi_single}.
Note that until point (A), all curves (for $k=1$) follow the same line as the unconstrained case.
However, once a community's ability to export excess generation is constrained at some time~$t$, it will begin to require larger operational capacity in order to achieve the same value of $\mathcal{X}_k$ as the unconstrained case.

Fig.~\ref{fig:max_chi_single} also shows the special case where a community is unable to export any excess power.\footnote{This scenario can be modeled by $\overline{P}^{\text{S}}_{k,t} = P^{\text{L}}_{k,t} - P^{\text{G}}_{k,t}$ for all times~$t$.}
In order to achieve a value of $\mathcal{X}_k = 1$ in this case, the community will need to invest in enough local generation capacity to completely meet its local demand at all times (i.e., the community will need to operate as a microgrid).
This situation corresponds to point (B) in Fig.~\ref{fig:max_chi_single}.
Furthermore, any power export constraints between the unconstrained case and the zero export case will fall between these curves (i.e., in the shaded region of Fig.~\ref{fig:max_chi_single}).

\subsection{Equity considerations in energyshed operations}
\label{sec:equity}

It is important to recognize that different communities will have different capabilities to invest in new energy assets.
Unless careful consideration is given to the diversity in resources that is inherent to different communities, some frameworks for energyshed operations can lead to inequitable outcomes for some communities.

For example, communities with the capability (more resources, larger budget) to invest in local generation resources the fastest will be able to achieve values of $\mathcal{X}_k \ge 1$ more quickly by exporting excess generation to less-privileged communities.
However, if each community makes independent decisions with the goal of maximizing its own $\mathcal{X}_k$, then as the least privileged communities invest in local generation resources, they may not have the opportunity to export excess power (since importing power would decrease $\mathcal{X}_k$ for surrounding communities).
Thus, communities which are the slowest to achieve their energyshed policy goals will have to invest disproportionately more resources to meet those goals. 
This motivates the need for a cooperative framework for energyshed operations, where the impact of decisions made by one community on neighboring communities is considered explicitly, in order to achieve better system-wide outcomes. 

\section{Framework for Energyshed Operations}
\label{sec:framework}

In this section, we explore how the technical definition for energysheds proposed in Sec.~\ref{sec:eshed_def} can be used to develop a mathematical framework for studying energysheds within the context of future power system operations.
We begin by considering a scenario in which energyshed policy requirements are already established and study how communities can meet these constraints at minimum cost.
Then, we propose methods for designing energyshed policies and explore tradeoffs between policy requirements and system-wide costs.

\subsection{Constraining the minimum local generation ratios}
\label{sec:case1}

Consider the case where each community~$k$ places a requirement (e.g., due to policy decisions or laws) on the minimum local generation ratio, denoted $\underline{\mathcal{X}}_k$, that it must achieve over the time interval $\mathcal{T}$. 
We take the perspective of a centralized power system operator who wishes to manage power flows within the network in order to minimize system costs, while adhering to energyshed policy requirements and without violating the limits of physical infrastructure.
This is encapsulated by the following optimization problem:
\begin{subequations}
\begin{align}
    \label{eq:obj_fun}
    \text{(P1)}\min_{P_{i,t}^{\text{S}+},P_{i,t}^{\text{S}-}}~&f_0(P_{i,t}^{\text{S}+},P_{i,t}^{\text{S}-}) , \\ 
    \text{s.t.}~~
    &P_{i,t}^{\text{G}} - P_{i,t}^\text{L} + P_{i,t}^{\text{S}} = \hspace{-0.75em}\sum_{j
    :(i,j)\in\mathcal{E}} \hspace{-0.5em}P_{ij,t}\,,~\forall i \in \mathcal{N}\,,~\forall t \in \mathcal{T}\,, \label{eq:pow_bal}\\
    &P_{ij,t} {x_{ij}} = {\theta_{i,t}-\theta_{j,t}}\,,~\forall (i,j) \in \mathcal{E}\,,~\forall t \in \mathcal{T}\,, \label{eq:power_flow}\\
    & \underline{P}_{ij,t} \le P_{ij,t} \le \overline{P}_{ij,t}\,,~\forall (i,j) \in \mathcal{E}\,,~\forall t \in \mathcal{T}\,, \label{eq:Pij_lim} \\
    & P_{i,t}^{\text{S}} = P_{i,t}^{\text{S}+} - P_{i,t}^{\text{S}-}\,,~\forall i \in \mathcal{N}\,,~\forall t \in \mathcal{T}\,, \label{eq:Ps_def_con}\\ 
    & 0 \le P_{i,t}^{\text{S}+} \le \overline{P}_{i,t}^{\text{S}+}\,,~~\forall i \in \mathcal{N}\,, ~~\forall t \in \mathcal{T}\,, \label{eq:Ps_plus_con}\\
    & 0 \le P_{i,t}^{\text{S}-} \le \overline{P}_{i,t}^{\text{S}-}\,,~~\forall i \in \mathcal{N}\,, ~~\forall t \in \mathcal{T}\,, \label{eq:Ps_minus_con}\\
    &\frac{\sum_{i\in\mathcal{N}_k}\sum_{t\in\mathcal{T}} (P^\text{G}_{i,t}+P^{\text{S}+}_{i,t})}{\sum_{i\in\mathcal{N}_k}\sum_{t\in\mathcal{T}} (P^\text{L}_{i,t}+P^{\text{S}-}_{i,t})} \ge \underline{\mathcal{X}}_k\,, ~~\forall k \in \mathcal{S}\,, \label{eq:chi_def_con} 
\end{align}
\end{subequations}
where the constraints~\eqref{eq:pow_bal}--\eqref{eq:power_flow} represent the DC power flow at each time~$t$, $x_{ij}$ is the reactance of line $(i,j)$, and $\theta_{i,t}$ (with $\theta_{1,t}=0$) is the voltage phase angle (reference angle) in radians of node~$i$.
Constraint~\eqref{eq:Pij_lim} captures flow limits of lines and transformers in the network.
The definition of $P^{\text{S}}_{i,t}$ is given by~\eqref{eq:Ps_def_con}, whereas constraints~\eqref{eq:Ps_plus_con}--\eqref{eq:Ps_minus_con} enforce bounds on $P^{\text{S}+}_{i,t}$ and $P^{\text{S}-}_{i,t}$.
Finally, the constraint~\eqref{eq:chi_def_con} captures the minimum local generation ratio requirement for each energyshed, based on the energyshed definition in~\eqref{eq:chi_def}, where $\mathcal{S}$ denotes the set of energyshed indices.

The objective function in~\eqref{eq:obj_fun} is characterized by a future power system, where the majority of generation comes from zero-marginal-cost renewable sources.
Thus, rather than operational costs based on fuel consumption and heat rate curves, we consider system costs that are driven by the required capacities of distributed energy resources~\cite{Hines2019}. 
Towards that goal, consider the following objective function:
\begin{equation} \label{eq:obj_fun_cap}
    f_0(P_{i,t}^{\text{S}+},P_{i,t}^{\text{S}-}) = \sum_{i\in \mathcal{N}} \alpha_i\left(\max_{t\in\mathcal{T}} \{P_{i,t}^{\text{S}+}\}\right)^2 + \beta_i \left(\max_{t\in\mathcal{T}} \{P_{i,t}^{\text{S}-}\}\right)^2\,,
\end{equation}
which captures this with quadratic functions of the capacities of $P^{\text{S}+}_{i,t}$ and $P^{\text{S}-}_{i,t}$ across all nodes in the system.
We chose to use a quadratic function of capacity costs because we assume that the incremental cost of adding capacity will increase with higher existing capacity in an area (e.g., due to different technologies, land-use constraints, and hosting capacity issues).
Furthermore, we introduce weights $\alpha_i$ and $\beta_i$ to account for the fact that the costs of installing additional capacity may be different at each node in the network. 

\begin{prop}
    If $f_0(P_{i,t}^{\text{S}+},P_{i,t}^{\text{S}-})$ is convex, then the optimization problem~(P1) is convex.
\end{prop}
\begin{proof}
    This is trivial since the denominator in~\eqref{eq:chi_def_con} is strictly positive by definition and $\underline{\mathcal{X}}_k$ is data.  
    Thus, constraint~\eqref{eq:chi_def_con} can be reformulated as a simple linear inequality. 
    \vspace{-0.5em}
\end{proof}
By solving (P1), we can determine the flexiblity required at each network bus in order to meet given energyshed policy requirements while minimizing system-wide costs.
Next, we explore how such policy requirements could be designed to ensure more equitable outcomes.

\subsection{Cooperative energyshed policy design}
\label{sec:case2}

In Sec.~\ref{sec:case1}, we assumed that the minimum local generation ratio for each energyshed was known based on policy requirements.
However, in this case, we investigate how policies can be designed to best enable communities to achieve their energy transition goals.
More specifically, policy makers may be interested in understanding what energyshed requirements are achievable given available resources. 
To this end, we consider the following \textit{max-min} optimization problem:
\begin{subequations}
\label{eq:opt_prob2}
\begin{align}
    \label{eq:obj_fun2}
    \text{(P2)}\max_{P_{i,t}^{\text{S}+},P_{i,t}^{\text{S}-},\mathcal{X}_k}~~& \min_{k\in\mathcal{S}}\{\mathcal{X}_k\}\,, \\ 
    \text{s.t.}~~
    &\eqref{eq:pow_bal}\text{--}\eqref{eq:Ps_minus_con}\,,\\
    &\frac{\sum_{i\in\mathcal{N}_k}\sum_{t\in\mathcal{T}} (P^\text{G}_{i,t}+P^{\text{S}+}_{i,t})}{\sum_{i\in\mathcal{N}_k}\sum_{t\in\mathcal{T}} (P^\text{L}_{i,t}+P^{\text{S}-}_{i,t})} = \mathcal{X}_k\,, ~\forall k \in \mathcal{S}\,. \label{eq:chi_def_con2} 
\end{align}
\end{subequations}

The goal of the objective function~\eqref{eq:obj_fun2} is to maximize the minimum value of $\mathcal{X}_k$ across all energysheds.
The solution of the optimization problem~(P2), denoted $\mathcal{X}^{*}_{\hat{k}}$, is the largest lower bound on $\mathcal{X}_k$ that all energysheds $k$ can achieve for a given base load and generation profile, line flow limits, and resource capacities, $\overline{P}^{\text{S}+}_{i,t}$ and $\overline{P}^{\text{S}-}_{i,t}$, at each node $i \in \mathcal{N}_k$ of each energyshed $k$.
We denote the energyshed achieving $\mathcal{X}^{*}_{\hat{k}}$ as energyshed $\hat{k}$.
Therefore, if policies are put in place that require any energyshed $\hat{k}$ to achieve $\underline{\mathcal{X}}_{\hat{k}} > \mathcal{X}^{*}_{\hat{k}}$, then these policies will not be achievable.
Also, note that the choice of objective function promotes other communities to actively support energyshed $\hat{k}$ in maximizing its local generation ratio. Thus, (P2) represents a cooperative approach to energyshed operations and informs energyshed policy design, in contrast to the competitive (myopic) approach discussed in Sec.~\ref{sec:equity}.

Of course, to find the optimal local generation ratio, denoted $\mathcal{X}_k^\ast$, for each energyshed $k$, we need to solve (P2). 
However, since $\mathcal{X}_k$ is a variable in the linear fractional equality constraint that is~\eqref{eq:chi_def_con2} (and not data as in (P1)), (P2) is clearly non-convex.\footnote{Note that multiplying both sides of~\eqref{eq:chi_def_con2} by the denominator of the left-hand side results in bi-linear terms (associated with multiplication of the variables $\mathcal{X}_k$ and $P^{\text{S}-}_{i,t}$) on the right-hand side.}
However, (P2) can be reformulated as a quasi-convex optimization problem for which standard methods exist for finding the globally optimal solutions.

\begin{prop}
    The non-convex optimization problem~(P2) can be solved to global optimality within arbitrary precision.
\end{prop}
\begin{proof}
Reformulate (P2) using the epigraph of~\eqref{eq:chi_def_con2}:
\begin{subequations}
\label{eq:opt_prob2_reform}
\begin{align}
    \label{eq:obj_fun2_reform}
    \text{(P3)}\max_{P_{i,t}^{\text{S}+},P_{i,t}^{\text{S}-},\tau}~~& \tau\,, \\ 
    \text{s.t.}~~
    &\eqref{eq:pow_bal}\text{--}\eqref{eq:Ps_minus_con}\,,\\
    &\frac{\sum_{i\in\mathcal{N}_k}\sum_{t\in\mathcal{T}} (P^\text{G}_{i,t}+P^{\text{S}+}_{i,t})}{\sum_{i\in\mathcal{N}_k}\sum_{t\in\mathcal{T}} (P^\text{L}_{i,t}+P^{\text{S}-}_{i,t})} \ge \tau\,, ~\forall k \in \mathcal{S}\,. \label{eq:chi_def_con2_reform} 
\end{align}
\end{subequations}
Since the left-hand side of~\eqref{eq:chi_def_con2} is a linear-fractional function, it is quasi-linear and, hence, quasi-concave, which means that its super-level sets are convex~\cite{boyd}. 

Thus, for a fixed value (e.g., guess) of $\tau$, the problem~(P3) represents a convex \textit{feasiblity problem}.
If (P3) is feasible for a given fixed value of $\tau$,  then we know that the optimal solution to (P2) satisfies 
$\mathcal{X}^{*}_{\hat{k}} \ge \tau$. 
Conversely, if~(P3) is infeasible for a given $\tau$, then $\mathcal{X}^{*}_{\hat{k}} < \tau$. Since the minimum $\mathcal{X}^{*}_{\hat{k}}$ is expected to reside in the compact interval $[0,1]$, we can apply the bisection method to solve (P2) to within $\epsilon$-optimality in  $\log_2 (1/\epsilon)$ iterations~\cite{boyd}.
For example, in 20~iterations of the bisection method applied to (P3), we can find a solution to (P2) that is within $\epsilon = 10^{-6}$ of the globally optimal solution. This completes the proof.


\end{proof}



There are several reasons for which it may be desired to extend the framework proposed in this subsection (Sec.~\ref{sec:case2}) to include more general objective functions. 
First, while~(P2) illuminates what energyshed policies are possible given available operational resources, it is agnostic to the costs of flexibility necessary to achieve such policies.
For example, policy makers may be interested in understanding not only if certain policies are technically feasible, but also the inherent tradeoffs between policy decisions and economic, environmental, and social costs.
Second, the solution of the optimization problem~(P2) is non-unique, since $\mathcal{X}_k$ variables that are not the minimum (i.e., $k \neq \hat{k}$) can take multiple values at optimality.
By adding additional terms to the objective function~\eqref{eq:obj_fun2}, we can attempt to reduce the number of non-unique solutions.
Thus, we introduce a generalized framework for cooperative energyshed policy design next.

\subsection{Generalized cooperative energyshed policy design}
Consider the following generalization of (P2):
\begin{subequations}
\label{eq:opt_prob3}
\begin{align}
    \label{eq:obj_fun3}
    \text{(P4)}~~\max_{P_{i,t}^{\text{S}+},P_{i,t}^{\text{S}-},\mathcal{X}_k}~~& \min_{k\in\mathcal{S}}\{\mathcal{X}_k\} - \frac{1}{\zeta}f_0(P_{i,t}^{\text{S}+},P_{i,t}^{\text{S}-})\,, \\ 
    \text{s.t.}~~
    &\eqref{eq:pow_bal}\text{--}\eqref{eq:Ps_minus_con}~\text{and}~\eqref{eq:chi_def_con2}\,.
\end{align}
\end{subequations}
Here, we introduce a generic cost function $f_0(P_{i,t}^{\text{S}+},P_{i,t}^{\text{S}-})$ and scalar weight, $\zeta$.
The function $f_0$ could consist of a number of system costs that we desire to minimize, including network losses, generation costs, or the quadratic costs on operational capacity $P^{\text{S}+}_{i,t}$ and $P^{\text{S}-}_{i,t}$ introduced in~\eqref{eq:obj_fun_cap}.
However, since quasi-concavity is not preserved under addition, the objective function in (P4) need not be quasi-concave
, even if $f_0$ is convex or quasi-convex. 
Thus, the bisection method used to iteratively solve (P2) cannot be directly applied in (P4). 
Nevertheless, we can leverage quasi-linearity and a similar reformulation as in Sec.~\ref{sec:case2} to develop methods for solving (P4) to global optimality.

\begin{prop}
    If $f_0$ is convex, then an $\epsilon$-optimal solution to the non-convex problem~(P4) can found by solving a sequence of convex problems of the form~(P1) by fixing $\tau$.
\end{prop}
\begin{proof}
First, we consider the following equivalent reformulation of~(P4) based on the epigraph of the first term:
\begin{subequations}
\begin{align}
    \label{eq:obj_fun3_reform}
    \text{(P5)}~~\max_{P_{i,t}^{\text{S}+},P_{i,t}^{\text{S}-},\tau}~~& \tau- \frac{1}{\zeta}f_0(P_{i,t}^{\text{S}+},P_{i,t}^{\text{S}-})\,, \\ 
    \text{s.t.}~~
    &\eqref{eq:pow_bal}\text{--}\eqref{eq:Ps_minus_con}\,,\\
    &\frac{\sum_{i\in\mathcal{N}_k}\sum_{t\in\mathcal{T}} (P^\text{G}_{i,t}+P^{\text{S}+}_{i,t})}{\sum_{i\in\mathcal{N}_k}\sum_{t\in\mathcal{T}} (P^\text{L}_{i,t}+P^{\text{S}-}_{i,t})} \ge \tau\,, ~\forall k \in \mathcal{S}\,. \label{eq:chi_def_con3_reform} 
\end{align}
\end{subequations}

Next, by treating $\tau$ as a parameter (i.e., data) rather than a variable, we derive the following parametric optimization problem from (P5)~\cite{still2018}:
\begin{subequations}
\begin{align}
    \label{eq:obj_fun3_reform2}
    \text{(P6)}~f(\tau):=\hspace{-0.5em}\max_{P_{i,t}^{\text{S}+},P_{i,t}^{\text{S}-}}~& \tau- \frac{1}{\zeta}f_0(P_{i,t}^{\text{S}+},P_{i,t}^{\text{S}-})\,, \\ 
    \text{s.t.}~
    &\eqref{eq:pow_bal}\text{--}\eqref{eq:Ps_minus_con}\,,\\
    &\frac{\sum_{i\in\mathcal{N}_k}\sum_{t\in\mathcal{T}} (P^\text{G}_{i,t}+P^{\text{S}+}_{i,t})}{\sum_{i\in\mathcal{N}_k}\sum_{t\in\mathcal{T}} (P^\text{L}_{i,t}+P^{\text{S}-}_{i,t})} \ge \tau\,, ~\forall k \in \mathcal{S}\,. \label{eq:chi_def_con3_reform2} 
\end{align}
\end{subequations}
Note that this parametric problem is in fact a special case of (P1), where we set $\underline{\mathcal{X}}_k = \tau$ for all energysheds $k\in\mathcal{S}$, and the objective function is scaled (by $1/\zeta$) and shifted by a fixed $\tau$.
Moreover, for a given value of $\tau$, we can evaluate the objective function~\eqref{eq:obj_fun3_reform2} as $f(\tau) = \tau-\frac{1}{\zeta}f_0(y_{\tau})$, where $y_{\tau}$ denotes the solution of a specific instance of the convex problem (P1) parameterized by $\tau$.
Thus, if we assume that the optimal value of $\tau \in [\mathcal{X}_\text{lb},\mathcal{X}_\text{ub}] = [0,1]$, then the global optimal solution of~(P4) can be found to within arbitrary precision $\epsilon$ by sweeping through this interval and evaluating $f(\tau)$.
\end{proof}

The mesh resolution used to sweep through values of $\tau$ is key for ensuring the global optimal is found.
Clearly, a finer mesh requires solving more instances of the convex problem~(P1), and thus is more computationally expensive.
If $f(\tau)$ is relatively smooth (does not change sharply for small changes in $\tau$), then this approach can be computationally tractable.
In the numerical case studies to follow, we will show that $f(\tau)$ is smooth for the specific choice of $f_0$ in~\eqref{eq:obj_fun_cap}.

In summary, we have developed a mathematical framework for studying energyshed policy design and operations in future power systems.
We showed that when policy requirements are fixed, system capacity costs can be minimized by solving the convex problem (P1).
We also proposed a cooperative policy design problem (P2) and extended it to include more general objective functions in (P4).
Finally, we showed that these non-convex problems can be solved to global optimality.
Next, we apply these results in illustrative examples.


\section{Numerical Case Study}
\label{sec:results}

In order to study the fundamental tradeoffs associated with energyshed policy decisions within future power system operations, we tested the proposed energyshed framework discussed in Sec.~\ref{sec:framework} on the IEEE 39-bus New England transmission network shown in Fig.~\ref{fig:ieee39bus}~\cite{matpower}.
\begin{figure}
    \centering
    \resizebox{0.45\textwidth}{!}{
    \input{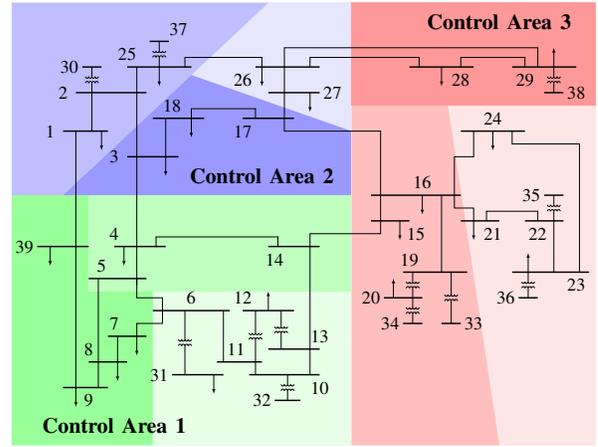}
    }
    \caption{IEEE 39-bus New England power system~\cite{matpower}. Different colors represent the three control areas, and shaded regions show energyshed boundaries for the medium aggregation case in Sec.~\ref{sec:case_setup}.}
    \label{fig:ieee39bus}
\end{figure}
The optimization problems (P1) and (P4) are implemented in Julia~(v1.9.3) using the JuMP~(v1.20.0) package~\cite{JuMP}, and are solved using IPOPT~(v1.6.1)~\cite{Ipopt}.


\subsection{Modifications to 39-bus system and case study setup}
\label{sec:case_setup}

Since we are considering a future power system where most generation comes from distributed and renewable sources, we have removed the conventional synchronous generator units from the 39-bus test network.
We also introduce solar DG at Buses 3, 4, 8, 12, 16, 20, 24, 25, 27, and 28.
This solar PV is treated as pre-existing and non-dispatchable generation (i.e., part of $P^{\text{G}}_{i,t}$) within the energyshed framework.
Load and existing distributed generation profiles over the day at each bus are based on the ACTIVSg time series data, and shown in Fig.~\ref{fig:load_der_profiles}~\cite{activgs_ts_data}.
We consider the set $\mathcal{T}$ as every hour over the course of one day.
Line flow limits for each transmission line, with $\underline{P}_{ij,t} = -\overline{P}_{ij,t}$, are identical to those in the MATPOWER case data for the 39-bus system~\cite{matpower}.
Unless otherwise noted, per unit quantities are provided on a 100~MVA system base.
\begin{figure}
    \centering
    \includegraphics[width=0.48\textwidth]{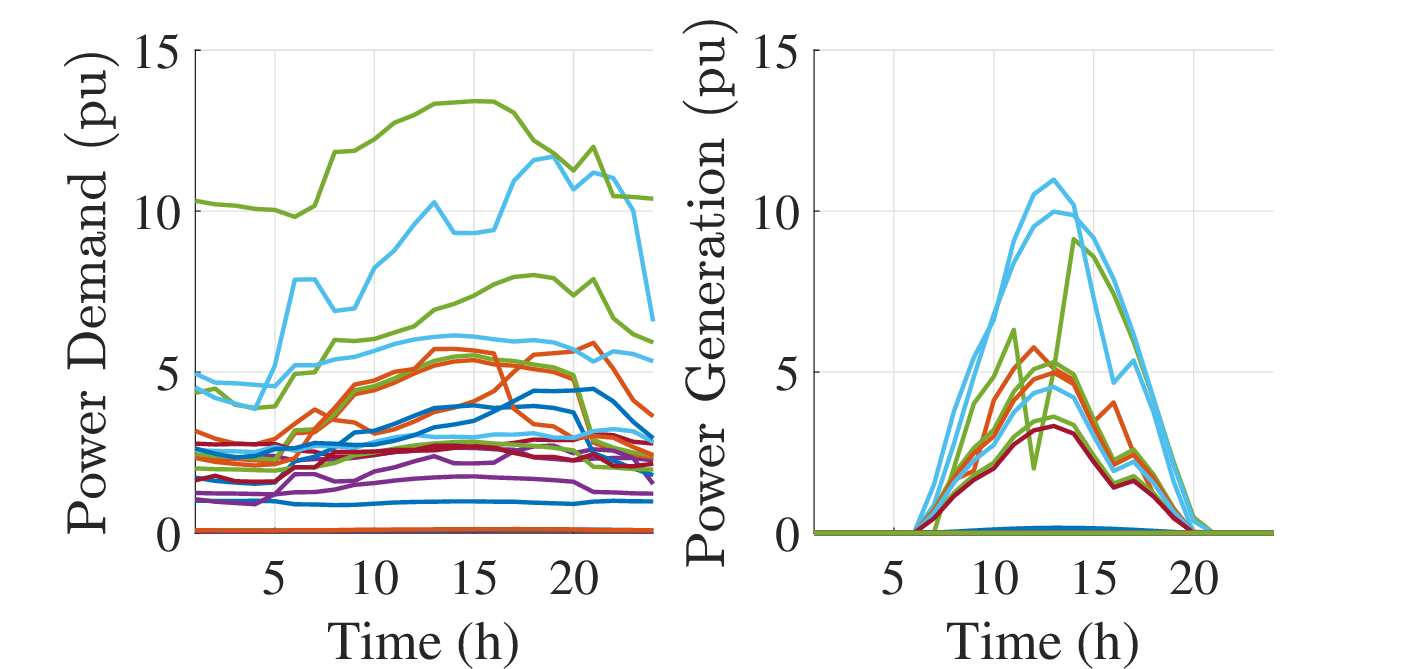}
    \caption{Hourly load and solar PV profiles at each load bus.}
    \label{fig:load_der_profiles}
\end{figure}

For simplicity, we assume that the net flexibility $P^{\text{S}}_{i,t}$ can only be modified at buses with existing load (i.e., $\overline{P}^{\text{S}+}_{i,t}=\overline{P}^{\text{S}-}_{i,t}=0$ if $\gamma_i = 0$).
The weights, $\alpha_i$ and $\beta_i$, associated with the quadratic capacity costs in~\eqref{eq:obj_fun_cap} for each of these buses are shown in Fig.~\ref{fig:weights_distr}.
\begin{figure}
    \centering
    \includegraphics[width=0.48\textwidth]{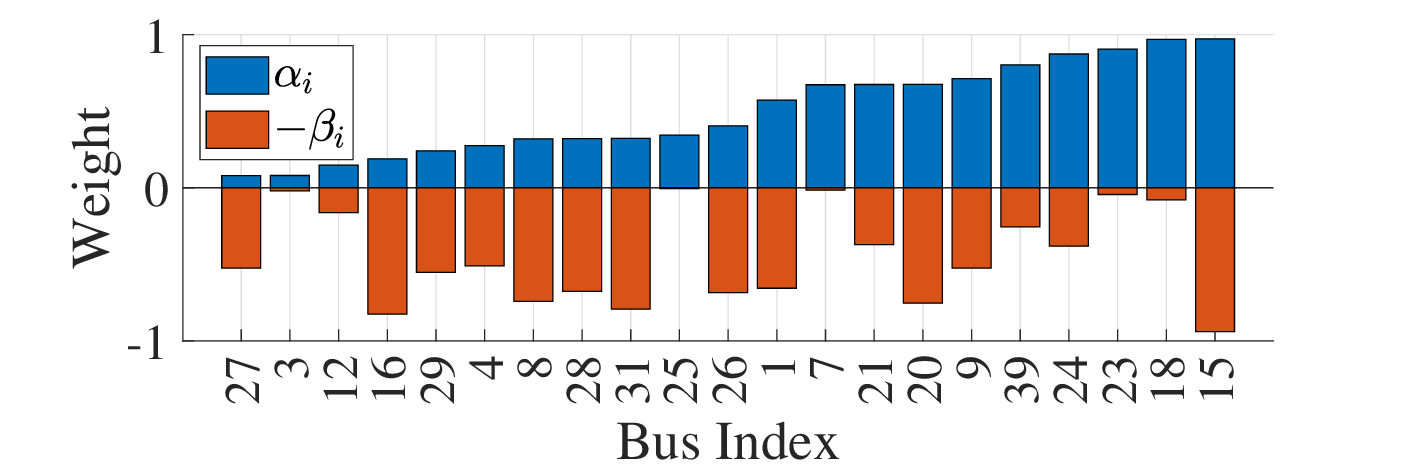}
    \caption{Capacity cost weights, $\alpha_i$ and $\beta_i$, for each load bus.}
    \label{fig:weights_distr}
\end{figure}

Finally, in order to study how the definition of energyshed geographic boundaries impacts energyshed policy design, we consider different configurations of the sets $\mathcal{N}_k$.
Specifically, we consider three cases with varying levels of spatial aggregation: 
(i) low aggregation, where each individual load bus is an energyshed (i.e., $\mathcal{N}_k = \{i\}$ for all $i$ such that $\gamma_i \neq 0$); 
(ii) medium aggregation, where energyshed boundaries are defined by the shaded regions shown in Fig.~\ref{fig:ieee39bus};
and (iii) high aggregation, where energyshed boundaries are defined by the three control areas in Fig.~\ref{fig:ieee39bus}.

\subsection{Impact of Local Generation Ratio Constraints}

In this section, we consider the low spatial aggregation case and solve (P1) under two scenarios of energyshed requirements.
First, we consider the situation where $\underline{\mathcal{X}}_k = 0$ for all energysheds, which corresponds to no energyshed requirements being enforced.
We treat this as a baseline case.
Next, we consider the situation where $\underline{\mathcal{X}}_k = 1$ for all energysheds.
Fig.~\ref{fig:cap_distr} shows the local generation ratio (LGR) and additional capacity for flexibility (both generation and demand) that minimize capacity costs for these two scenarios.
Note that in the baseline case, zero additional demand capacity is added ($\max_{t\in\mathcal{T}} \{P_{i,t}^{\text{S}-}\} = 0$); however, some additional generation capacity is required in order to meet existing demand (even though no energyshed policy requirements are in place).
Furthermore, when $\underline{\mathcal{X}}_k = 0$, some of the energysheds (e.g., those with smaller $\alpha_i$) have $\mathcal{X}_k > 1$, while other energysheds have $\mathcal{X}_k < 1$.
Clearly, when the energyshed policy constraint requires $\underline{\mathcal{X}}_k = 1$, all energysheds have $\mathcal{X}_k = 1$.
To achieve this, we can see that energysheds with $\mathcal{X}_k > 1$ in the baseline case (e.g., Buses 3, 27, 12) decrease their generation capacity, while energysheds with $\mathcal{X}_k < 1$ in the baseline case (e.g., Bus 39) increase their generation capacity.
Furthermore, some energysheds with $\mathcal{X}_k > 1$ in the baseline case also increase their demand capacity, which allows them to import from surrounding energysheds with expensive generation capacity costs during periods of excess generation.
\begin{figure}
    \centering
    \begin{tikzpicture}
    \node[inner sep=0pt] at (0.0,0.0)
    {\includegraphics[width=0.48\textwidth]{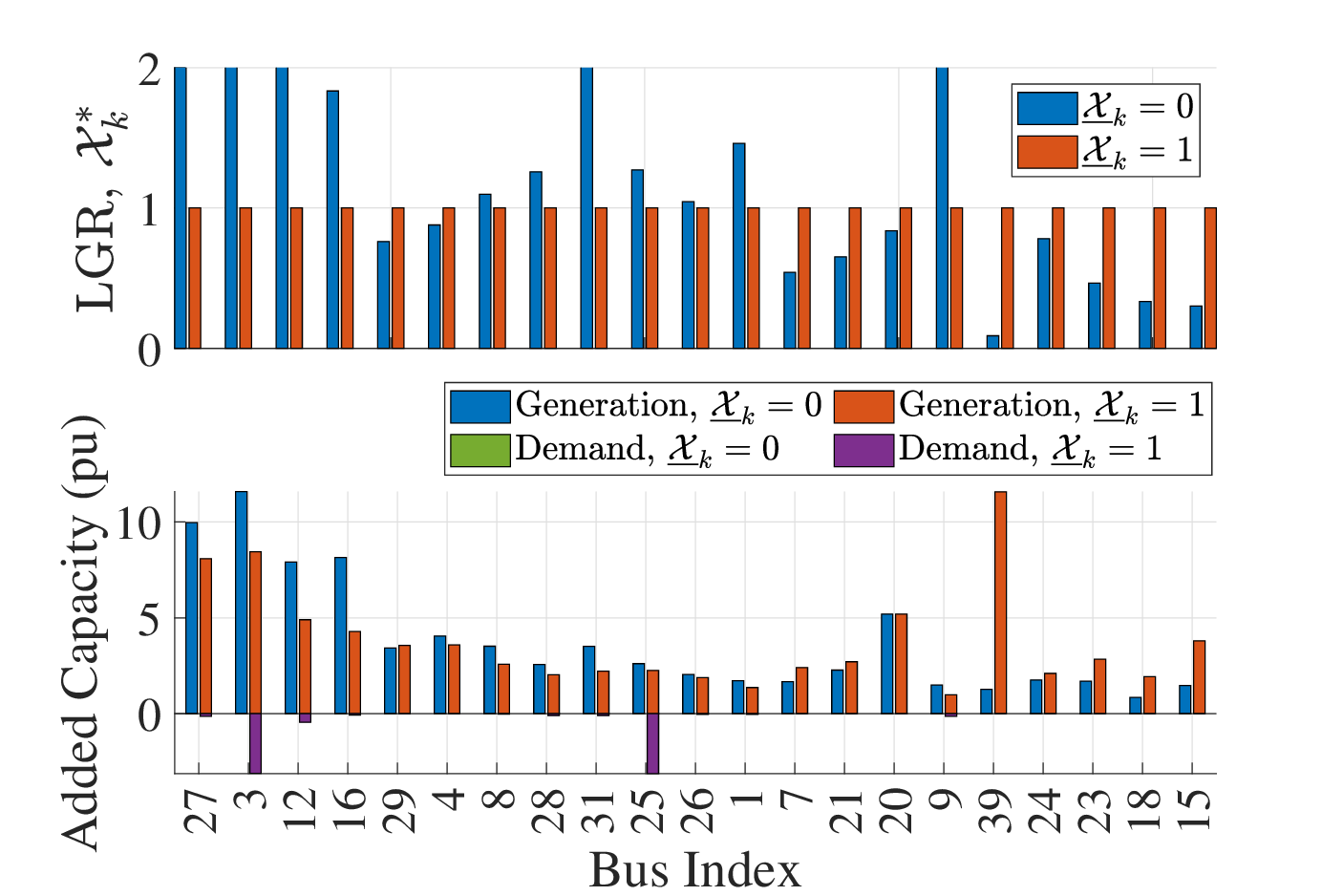}};
    \node[] () at (-4.25,1.5) {(a)};
    \node[] () at (-4.25,-1.25) {(b)};
    \end{tikzpicture}
    \caption{Distribution of local generation ratio (a) and additional flexibility capacity (b) across load buses under different energyshed policy constraints. Buses are arranged from left to right in order of increasing generation capacity costs ($\alpha_i$). Note that $\mathcal{X}^*_k$ is larger than $2.0$ for some buses, but we truncate the plot at this value for improved legibility.}
    \label{fig:cap_distr}
\end{figure}

\subsection{Case study on cooperative energyshed policy design}

Next, we study energyshed policy design by solving (P4) for the different spatial aggregation cases.
We begin by showing an example of solving (P4) to global optimality for the medium aggregration case.
Fig.~\ref{fig:chi_sweep} shows the objective function $f(\tau)$ as $\tau$ is swept through the interval $[0,1]$ in steps of $0.01$, for different values of the weight parameter $\zeta$.
For this choice of $f_0$, the objective function value of~\eqref{eq:obj_fun3_reform2} as a function of $\tau$ is quite smooth and unimodal; however, note that this may not hold for general $f_0$.
The solid dots in Fig.~\ref{fig:chi_sweep} mark the solution found using this sweeping approach, which is clearly optimal.
\begin{figure}
    \centering
    \includegraphics[width=0.48\textwidth]{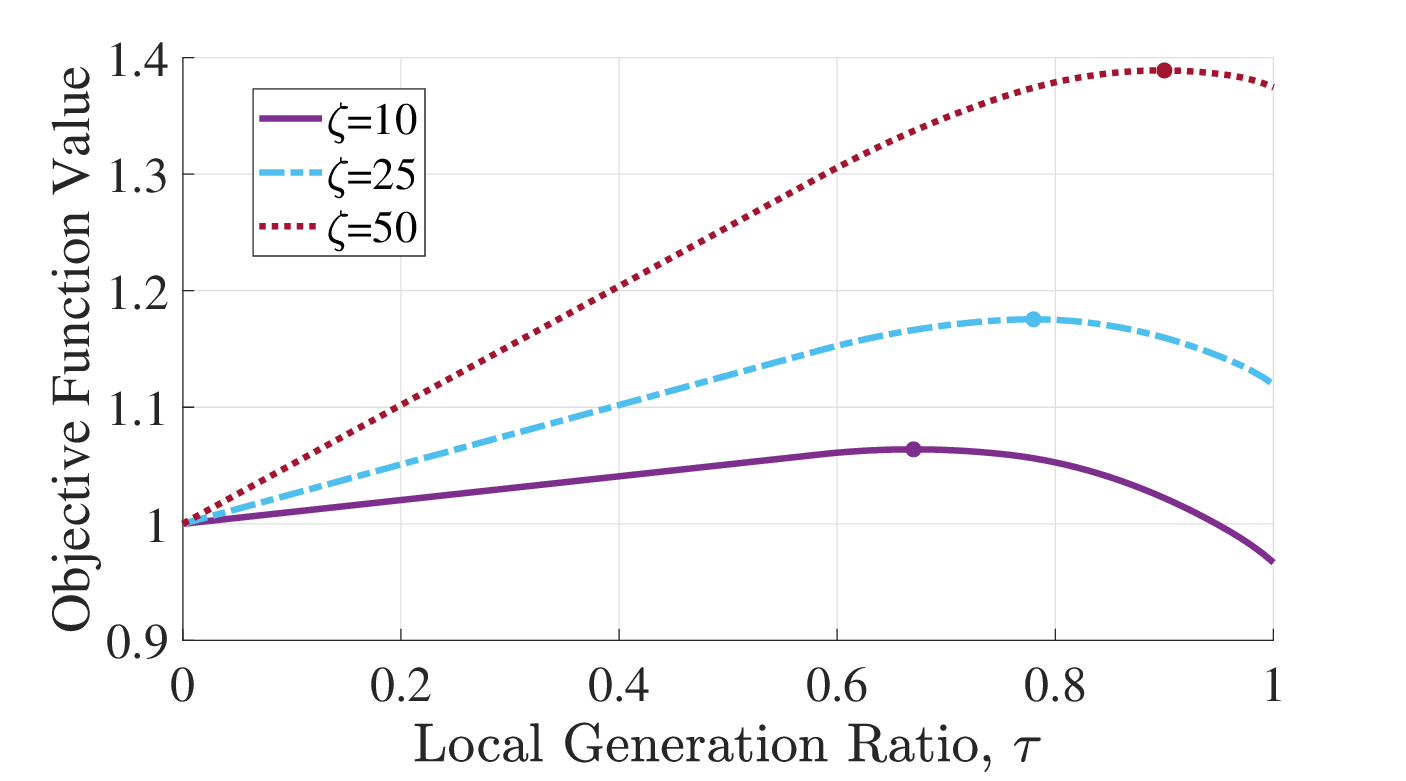}
    \caption{Value of objective function $f(\tau)$ as local generation ratio is swept in the interval $[0,1]$, for three different values of $\zeta$, in the medium spatial aggregation case. Objective function values are normalized with respect to the baseline case.}
    \label{fig:chi_sweep}
\end{figure}

The Pareto-optimal front, parameterized by $\zeta$, of (P4)  for the three spatial aggregation cases is shown in Fig.~\ref{fig:pareto_front}.
These curves present tradeoffs between the minimum local generation ratio required by energyshed policy decisions and system-wide costs of increasing capacity for operational flexibility.
It is observed that as policy decisions require communities to meet more demand from local generation, system capacity costs increase.
Moreover, the larger energyshed boundaries become, the less expensive it is to achieve policy objectives for those energysheds.
For example, a $102\%$ increase in capacity costs over the baseline is required to reach $\underline{\mathcal{X}}_k = 1$ when each load bus is an energyshed, whereas only a $13.5\%$ increase in costs occurs when energysheds are defined by the boundaries in Fig.~\ref{fig:ieee39bus}.
Furthermore, when energysheds are based on the control areas of the 39-bus system (high aggregation case), it is possible to achieve $\underline{\mathcal{X}}_k = 1$ without additional capacity costs over the baseline case (i.e. the Pareto-optimal front is a single point).
\begin{figure}
    \centering
    \includegraphics[width=0.48\textwidth]{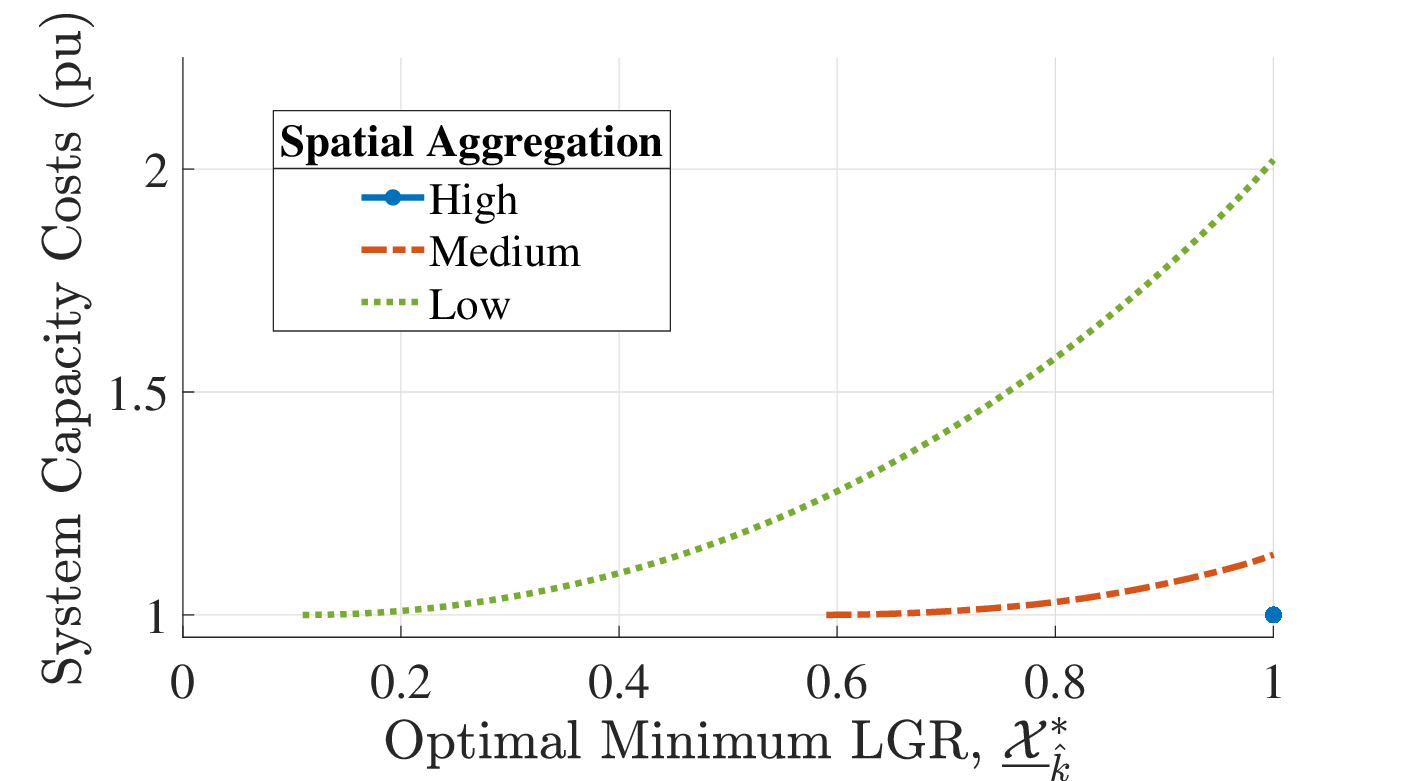}
    \caption{Pareto-optimal front of (P4) for different energyshed boundaries. Capacity costs are normalized with respect to the baseline case ($\underline{\mathcal{X}}_k = 0$).}
    \label{fig:pareto_front}
\end{figure}

The scalar weight $\zeta$ plays an important role in balancing the tradeoff between capacity costs and minimum local generation ratio requirements.
In particular, policy makers may be interested in placing a value on a specific energyshed policy requirement in terms of system costs.
Fig.~\ref{fig:scaling_vs_chi} shows the optimal energyshed policy requirement obtained from solving (P4) for different values of $\zeta$.
It is envisioned that these curves could be useful in helping policy makers assign value (or prices) to local generation ratio requirements relative to other system costs.
\begin{figure}
    \centering
    \includegraphics[width=0.48\textwidth]{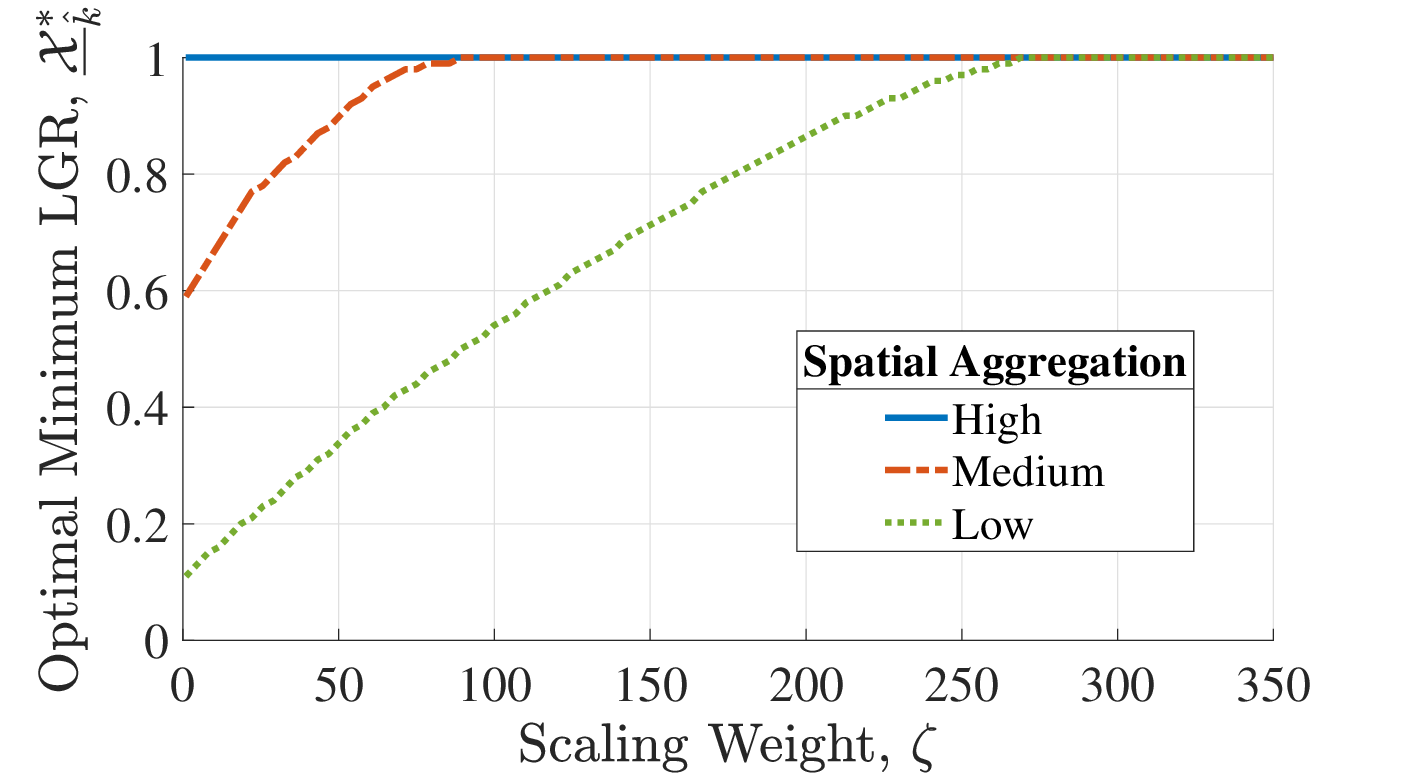}
    \caption{Optimal energyshed policy requirement as a function of weight parameter $\zeta$, for different spatial aggregation cases.}
    \label{fig:scaling_vs_chi}
\end{figure}


\section{Conclusion}
\label{sec:conclusion}

In this paper, we proposed a mathematical definition for energysheds, and studied the fundamental tradeoffs associated with energyshed policy decisions and operations within the context of future electric power systems.
We also explored how interactions between energyshed operations across different communities can potentially lead to inequitable outcomes, and introduced a framework for cooperative energyshed policy design.
Theoretical insights into energyshed concepts as well as a numerical case study were also presented.

There are numerous potential avenues for future research in energyshed thinking.
First, the proposed framework could be extended to consider AC network analysis (i.e., system voltages, reactive power flows, and line losses) or multi-energy systems (e.g., district heating or transportation systems), rather than only considering electric infrastructure.
Future work could also study how investments in energy efficiency, which would enable decreasing of local demand, may impact energyshed policy design.
An investigation of how different temporal energyshed policy requirements (e.g., enforcing ratios over every second, hour, day, or year) effect system costs would also be of interest.
Additionally, further analysis into computationally efficient methods for solving (P4), such as conditions under which it is quasi-convex, would enable improvements in scalability to larger systems.
Finally, the application of the proposed framework to large regions with realistic data, and consideration of economic, environmental, and social aspects of community energy operations, planning, and investment would be valuable. 


\addtolength{\textheight}{-12cm}   





\section*{ACKNOWLEDGMENT}

The authors greatly appreciate numerous team discussions about energysheds with colleagues at UVM, including Jeff Marshall, Jon Erickson, Greg Rowangould, Dana Rowangould, Bindu Panikkar, Hamid Ossareh, Eric Seegerstrom, Emmanuel Badmus, and Omid Mokhtari, as well as utility partners at Green Mountain Power (GMP), Vermont Electric Cooperative (VEC), Stowe Electric Department, and Vermont Electric Company (VELCO), and helpful feedback from Lorenzo Kristov.


\bibliographystyle{IEEEtran}
\bibliography{IEEEabrv,references}

\end{document}